\RequirePackage[l2tabu,orthodox]{nag}
\documentclass[11pt,letterpaper]{article}

\usepackage{wrapfig}

\usepackage{amsthm,amssymb,amsmath}
\usepackage{mathtools}
\usepackage{xspace,enumerate}
\newtheorem{theorem}{Theorem}[section]
\newtheorem*{theorem*}{Theorem}

\newtheorem*{proposition*}{Proposition}
\newtheorem{lemma}[theorem]{Lemma}
\newtheorem*{lemma*}{Lemma}
\newtheorem{corollary}[theorem]{Corollary}
\newtheorem*{conjecture*}{Conjecture}
\newtheorem{fact}[theorem]{Fact}
\newtheorem*{fact*}{Fact}

\newtheorem*{hypothesis*}{Hypothesis}

\theoremstyle{definition}
\newtheorem{definition}[theorem]{Definition}
\newtheorem*{definition*}{Definition}

\newtheorem{algorithm}[theorem]{Algorithm}

\theoremstyle{remark}

\newtheorem*{claim*}{Claim}
\newtheorem{remark}[theorem]{Remark}
\newtheorem*{remark*}{Remark}

\newtheorem*{observation*}{Observation}

\usepackage[
letterpaper,
top=1in,
bottom=1in,
left=1in,
right=1in]{geometry}

\usepackage{microtype}
\usepackage[
pagebackref,
colorlinks=true,
urlcolor=blue,
linkcolor=blue,
citecolor=blue,
]{hyperref}
\usepackage[capitalise,nameinlink]{cleveref}
\crefname{lemma}{Lemma}{Lemmas}
\crefname{fact}{Fact}{Facts}
\crefname{theorem}{Theorem}{Theorems}
\crefname{corollary}{Corollary}{Corollaries}
\crefname{claim}{Claim}{Claims}
\crefname{example}{Example}{Examples}
\crefname{algorithm}{Algorithm}{Algorithms}
\crefname{problem}{Problem}{Problems}
\crefname{definition}{Definition}{Definitions}
\usepackage{paralist}
\usepackage{turnstile}
\usepackage{mdframed}
\usepackage{tikz}
\usepackage{caption}
\DeclareCaptionType{Algorithm}
\usepackage{newfloat}

\newcommand{\pmn}[1]{\left\|#1\right\|_{\infty \to 1}}
\newcommand{\Authornote}[2]{\textcolor{blue}{#1: #2}}
\newcommand{\Authornotecolored}[3]{}
\newcommand{\Authorcomment}[2]{}
\newcommand{\Authorfnote}[2]{}

\newcommand{\Inote}{\Authornotecolored{ForestGreen}{I}}

\definecolor{forestgreen(traditional)}{rgb}{0.0, 0.27, 0.13}

\usepackage{boxedminipage}
\newcommand{\paren}[1]{(#1)}
\newcommand{\Paren}[1]{\left(#1\right)}

\newcommand{\Abs}[1]{\left\lvert#1\right\rvert}

\newcommand{\set}[1]{\{#1\}}

\newcommand{\norm}[1]{\lVert#1\rVert}
\newcommand{\Norm}[1]{\Bigl\lVert#1\Bigr\rVert}
\newcommand{\bignorm}[1]{\big\lVert#1\big\rVert}
\newcommand{\Bignorm}[1]{\Big\lVert#1\Big\rVert}

\newcommand{\Esymb}{\mathbb{E}}
\newcommand{\Psymb}{\mathbb{P}}
\newcommand{\Vsymb}{\mathbb{V}}
\DeclareMathOperator*{\E}{\Esymb}

\DeclareMathOperator*{\ProbOp}{\Psymb}
\renewcommand{\Pr}{\ProbOp}

\newcommand{\tensor}{\otimes}

\newcommand{\textparen}[1]{\text{(#1)}}

\newcommand{\mper}{\,.}
\newcommand{\mcom}{\,,}
\newcommand\bdot\bullet

\DeclareMathOperator{\sdp}{sdp}
\DeclareMathOperator{\val}{\mathrm{val}}

\DeclareMathOperator{\poly}{poly}

\DeclareMathOperator{\polylog}{polylog}

\newcommand{\Lovasz}{Lov\'asz\xspace}

\newcommand{\Hoelder}{H\"{o}lder\xspace}

\newcommand{\Z}{\mathbb Z}

\newcommand{\R}{\mathbb R}

\newcommand{\varR}{\Re}

\newcommand{\problemmacro}[1]{\texorpdfstring{\textup{\textsc{#1}}}{#1}\xspace}
\newcommand{\pnum}[1]{{\footnotesize #1}}

\newcommand{\wt}{\mathsf{wt}}

\renewcommand{\S}{\mathbb{S}}

\newcommand{\PTIME}{\mathrm{P}}
\newcommand{\NP}{\mathrm{NP}}
\newcommand{\CSP}{\mathrm{CSP}}

\newcommand{\calP}{\mathcal{P}}
\newcommand{\calE}{\mathcal{E}}
\newcommand{\calI}{\mathcal{I}}
\newcommand{\avg}{\mathrm{avg}}
\newcommand{\Opt}{\mathrm{opt}}
\newcommand{\dtv}{d_{\mathrm{TV}}}
\newcommand{\cmplx}{\mathcal{C}}

\renewcommand{\leq}{\leqslant}
\renewcommand{\le}{\leqslant}
\renewcommand{\geq}{\geqslant}
\renewcommand{\ge}{\geqslant}
\let\epsilon=\varepsilon
\numberwithin{equation}{section}
\newcommand\MYcurrentlabel{xxx}
\newcommand{\MYstore}[2]{%
  \global\expandafter \def \csname MYMEMORY #1 \endcsname{#2}%
}
\newcommand{\MYload}[1]{%
  \csname MYMEMORY #1 \endcsname%
}
\newcommand{\MYnewlabel}[1]{%
  \renewcommand\MYcurrentlabel{#1}%
  \MYoldlabel{#1}%
}
\newcommand{\MYdummylabel}[1]{}
\newcommand{\torestate}[1]{%
  \let\MYoldlabel\label%
  \let\label\MYnewlabel%
  #1%
  \MYstore{\MYcurrentlabel}{#1}%
  \let\label\MYoldlabel%
}
\newcommand{\restatetheorem}[1]{%
  \let\MYoldlabel\label
  \let\label\MYdummylabel
  \begin{theorem*}[Restatement of \cref{#1}]
    \MYload{#1}
  \end{theorem*}
  \let\label\MYoldlabel
}
\newcommand{\restatelemma}[1]{%
  \let\MYoldlabel\label
  \let\label\MYdummylabel
  \begin{lemma*}[Restatement of \cref{#1}]
    \MYload{#1}
  \end{lemma*}
  \let\label\MYoldlabel
}
\newcommand{\restateprop}[1]{%
  \let\MYoldlabel\label
  \let\label\MYdummylabel
  \begin{proposition*}[Restatement of \cref{#1}]
    \MYload{#1}
  \end{proposition*}
  \let\label\MYoldlabel
}
\newcommand{\restatefact}[1]{%
  \let\MYoldlabel\label
  \let\label\MYdummylabel
  \begin{fact*}[Restatement of \prettyref{#1}]
    \MYload{#1}
  \end{fact*}
  \let\label\MYoldlabel
}
\newcommand{\restate}[1]{%
  \let\MYoldlabel\label
  \let\label\MYdummylabel
  \MYload{#1}
  \let\label\MYoldlabel
}

\newcommand{\eps}{\epsilon}

\allowdisplaybreaks
\newcommand*{\on}{\{\pm 1\}}

\usepackage{todonotes}

\def\norm#1{\left\| #1 \right\|}

\title{Strongly refuting all semi-random Boolean CSPs}

\author{Jackson Abascal\thanks{Research supported in part by NSF grant CCF-1563742. {\tt jabascal@cs.cmu.edu}} \and Venkatesan Guruswami\thanks{Research supported in part by NSF grant CCF-1908125. {\tt venkatg@cs.cmu.edu}} \and Pravesh K. Kothari\thanks{{\tt praveshk@cs.cmu.edu}}}

\date{Computer Science Department \\ Carnegie Mellon University \\ Pittsburgh, PA 15213}
\begin{document}


\maketitle
\thispagestyle{empty} 



\begin{abstract}
We give an efficient algorithm to strongly refute \emph{semi-random} instances of all Boolean constraint satisfaction problems. The number of constraints required by our algorithm matches (up to polylogarithmic factors) the best known bounds for efficient refutation of fully random instances. Our main technical contribution is an algorithm to strongly refute semi-random instances of the Boolean $k$-XOR problem on $n$ variables that have $\widetilde{O}(n^{k/2})$ constraints. (In a semi-random $k$-XOR instance, the equations can be arbitrary and only the right hand sides are random.)

\smallskip
One of our key insights is to identify a simple combinatorial property of random XOR instances that makes spectral refutation work. Our approach involves taking an instance that does not satisfy this property (i.e., is \emph{not} pseudorandom) and reducing it to a partitioned collection of $2$-XOR instances. We analyze these subinstances using a carefully chosen quadratic form as proxy, which in turn is bounded via a combination of spectral methods and semidefinite programming. The analysis of our spectral bounds relies only on an off-the-shelf matrix Bernstein inequality. Even for the purely random case, this leads to a shorter proof compared to the ones in the literature that rely on problem-specific trace-moment computations.

\end{abstract}

\clearpage


  \microtypesetup{protrusion=false}
  \tableofcontents{}
  \thispagestyle{empty}
  \microtypesetup{protrusion=true}

\clearpage

\pagestyle{plain}
\setcounter{page}{1}

\def\biasval{\text{bias}}
\section{Introduction}  
\label{sec:intro}

The study of random constraint satisfaction problems (CSPs) is a major thrust in the theory of average-case algorithm design and complexity. Investigations into this problem are central to disparate areas including cryptography~\cite{ABW10}, proof complexity~\cite{BB02}, hardness of approximation~\cite{MR2121179-Feige02}, learning theory~\cite{DBLP:conf/stoc/DanielyLS14}, SAT-solving~\cite{SAT}, statistical physics~\cite{CLP02}, and complexity theory~\cite{BKS13}.

For a predicate $P : \{0,1\}^k \to \{0,1\}$, an instance of $\CSP(P)$ consists of a set of Boolean variables and a collection of constraints each of which applies $P$ to a tuple of $k$ literals (variables or their negations). The goal is to understand the maximum fraction of constraints that can be satisfied by any Boolean assignment to the variables. In the \emph{fully random} model, the $k$-tuples and the literal patterns are chosen independently and uniformly at random. In the \emph{semi-random} model, the $k$-tuples are arbitrary (i.e., ``worst-case'') and only the negation patterns are chosen randomly. When the number $m$ of constraints is much larger than the number $n$ of variables, such instances are unsatisfiable with high probability. The algorithmic goal is to find an efficient algorithm that, for most instances, finds a \emph{refutation}---an efficiently verifiable certificate that no assignment can satisfy \emph{all} constraints (\emph{weak} refutation). More stringently, we can ask for a \emph{strong refutation} that, for some absolute constant $\delta > 0$,  the instance is not ``$(1-\delta)$-satisfiable" in that no assignment can satisfy even a fraction
$(1-\delta)$ fraction of constraints.

The question of refuting fully-random CSPs was first considered in the seminal work of Feige~\cite{MR2121179-Feige02}. Recently, this effort has had a series of exciting developments giving both new efficient algorithms~\cite{AOW15,RRS16,DBLP:conf/approx/BhattiproluGL17} and nearly matching lower-bounds~\cite{BCK15,KMOW16,FPV15} in several strong models of computation. 

The focus of this work is the more challenging \emph{semi-random} CSP model. The semi-random (and the related, ``smoothed'') CSP model was formalized by Feige~\cite{Fei07} following a long line of work~\cite{blum-spencer,FK01} on studying semi-random instances of natural problems. In that work, Feige gave a \emph{weak} refutation algorithm for semi-random $3$-SAT and noted: \emph{“Our motivation for studying semi-random models is that they are more challenging than random models, lead to more robust algorithms, and lead to a better understanding of the borderline between difficult and easy instances.”} For an introduction and recent survey of semi-random models, see the chapter~\cite{bwca-semirandom} in ``Beyond Worst-case Analysis of Algorithms."

The question of semi-random refutation was revived recently in the work of Allen, O'Donnell and Witmer~\cite{AOW15} who gave state-of-the-art efficient strong refutation algorithms for fully random CSPs and explicitly asked whether their algorithm could extend to the more challenging semi-random setting. This direction was further fleshed out in Witmer's PhD thesis~\cite{Witmer17} (see Chapter 5). Two years later, a sequence of exciting works~\cite{DBLP:journals/iacr/LinT17,DBLP:conf/eurocrypt/Lin16,DBLP:conf/focs/LinV16,DBLP:conf/crypto/Lin17,cryptoeprint:2016:1097,DBLP:journals/iacr/GayJLS20} in cryptography related this question to the security of arbitrary instantiations of Goldreich's local pseudorandom generator~\cite{Gol00} and showed that it implies indistinguishability obfuscation under standard assumptions. In their work on refuting hardness assumptions in this context, \cite{DBLP:conf/eurocrypt/BarakBKK18} note that better algorithms for semi-random refutation imply stronger attacks on the candidate local pseudo-random generators in cryptography. 

\subsection{Our results}
In this work, we resolve the question of semi-random CSP refutation by giving polynomial time algorithms that match the guarantees on fully random CSP instances. 

\begin{theorem} \label{thm:main-general-CSPs}

Fix a predicate $P : \{0,1\}^k \to \{0,1\}$. Let $r(P) = \E_x [P(x)]$ over a uniformly random $x \in \{0,1\}^k$, let $d(P)$ be the degree of $P$ as a multilinear polynomial, and let $i(P)$ be the largest integer $t$ for which $P^{-1}(1)$ supports a $t$-wise uniform distribution, i.e., a distribution whose joint marginals on every $t$ variables is uniform.

There is a polynomial time algorithm that given an $n$-variable semi-random instance of $\CSP(P)$ with $m \ge n^{(i(P)+1)/2} \mathrm{poly}(\log n)$ constraints, with high probability certifies that it is not $(1-\delta)$-satisfiable for some $\delta = \Omega_k(1)$. 

Further, when the number of constraints exceeds $ n^{d(P)/2} \mathrm{poly}((\log n)/\eps)$, the algorithm certifies that the instance is not $(r(P)+\epsilon)$-satisfiable  (the ``right'' bound).

\end{theorem}

For the significance of the quantity $i(P)$ in the context of CSP refutation, see~\cite{AOW15} who showed strong refutation of \emph{random} CSPs for a similar number of constraints. Our work thus extends the state-of-the-art in CSP refutation from random to semi-random CSPs.

\medskip\noindent\textbf{It's all about XOR.} 
The question of refuting arbitrary Boolean CSPs immediately reduces to refuting the specific $k$-XOR CSPs\footnote{In $k$-XOR CSPs, the predicate $P$ computes the XOR of the $k$ input bits.} for all $k$ via Feige's~\cite{MR2121179-Feige02} ``XOR principle'' formalized in the work of~\cite{COCF10} (see also, Theorem 4.9 in~\cite{AOW15}). This involves writing the Boolean predicate as a multilinear polynomial (its discrete Fourier transform) and refuting each CSP instance corresponding to a term in this expansion. For even $k$, Witmer and Feige~\cite{Witmer17} observed that an idea from~\cite{Fei07} already gives a strong semi-random refutation algorithm (via ``discrepancy'') that matches the best-known guarantees in the fully-random model. For odd $k$, however, the best known efficient algorithm succeeds only when the number of constraints is $\tilde{\Omega}(n^{\lceil k/2 \rceil})$ which is larger by a factor of $\sqrt{n}$ when compared to the fully random case. Our main technical contribution is to resolve this deficiency for odd $k$:

 \begin{theorem}[Main, detailed version appears  as Theorem~\ref{kxor_refutation}] \label{thm:main-odd-xor}
\label{thm:main-intro}
There is a polynomial time algorithm that given an $n$-variable semi-random $k$-XOR instance for odd $k$ with $m \ge n^{k/2} \mathrm{poly}((\log n)/\eps)$ equations, with high probability certifies that it is not $(1/2+\eps)$-satisfiable.
\end{theorem}

A priori, the issue of odd vs even arity XOR might appear like a mere technical annoyance. However, even for the fully random case, strongly refuting odd-arity XOR turns out to be trickier\footnote{see discussion following Theorem 2.1 on Page 4 of~\cite{AOW15}.} and was resolved only in the recent work of Barak and Moitra~\cite{BM15} by means of an intense argument based on the trace-moment method. In his work on \emph{weak} refutation of semi-random 3-SAT, Feige manages to skirt this issue by observing that proving a slightly non-trivial (but still $1-o_n(1)$) upper-bound on the value of a semi-random $3$-XOR instance is enough. He accomplishes this goal by a sophisticated combinatorial argument (that improved on a work of Naor and Verstræte~\cite{MR2399017-Naor08}) based on the existence of small \emph{even covers} in sufficiently large hypergraphs. 

To obtain the strong refutation guaranteed by our work, we replace this combinatorial argument by a careful combination of spectral bounds and semidefinite programming. As explained in more detail in Section~\ref{sec:overview}, we identify a certain pseudorandomness property (which holds for random $k$-XOR instances) under which strong refutation follows from spectral bounds of certain matrices that are sums of independent random matrices. To tackle the semi-random case, we decompose the instance into various structured sub-instances and associated quadratic forms whose sum bounds the optimum of the CSP instance. While the spectral norm of the matrix associated with a sub-instance might be large, for the CSP refutation we only care about vectors with commensurately smaller effective $\ell_2$-norm. 
We also need to tackle variables which participate in too many constraints separately, and we do this by forking off an auxiliary $2$-XOR instance which is refuted using the Grothendieck inequality.

 Our work leads to many natural open questions. 
 An important one is whether we can remove the polylogarithmic multiplicative factors in Theorem~\ref{thm:main-intro} and strongly refute with $O(n^{k/2})$ constraints. We note that for the case of odd $k$, this is open even in the case of random $k$-XOR.  While this might seem like a rather technical question, it is interesting because a resolution would likely call for more sophisticated random matrix theory methods. In this regard, the work of \cite{DMOSS19} on random NAE-3SAT might be worth highlighting.

\subsection{Related works}
Here, we give a brief survey of the long line of work on average-case complexity of CSPs. 

There is  a long line of work on studying algorithms and lower-bounds for refuting/solving random CSPs and connections to other computational problems. We only include the highlights here and point the reader to~\cite{KMOW16} for a more extensive survey. Feige~\cite{MR2121179-Feige02} made a fruitful connection between average-case complexity of random CSPs and hardness of approximation and formulated his ``R3SAT'' hypothesis that conjectures that there's no polynomial time algorithm for strong refutation of $3$-SAT with $m = Cn$ constraints for some $C \gg 4.27$. This and its generalization to other CSPs have served as starting points of reductions that have yielded hardness results in several domains including cryptography, learning theory, and game theory~\cite{Fei02,Ale03,BKP04, DFHS06, Bri08, AGT12,DLS13,DBLP:conf/stoc/DanielyLS14,Dan15,ABW10,App13b}. Very recently, variants of the assumptions were used to give candidate constructions of indistinguishability obfuscation in cryptography~\cite{DBLP:journals/iacr/LinT17}.

Coja-Oghlan~\cite{COGL07} gave the first $\delta$ (i.e. certifying a bound of $1-\delta$ for some fixed constant $\delta > 0$) refutation algorithms for $3$ and $4$-SAT. Early works used combinatorial certificates of unsatisfiability for random 3SAT instances \cite{Fu96,DBLP:journals/siamcomp/BeameKPS02}. Goerdt and Krivilevich~\cite{GK01} were the first to examine spectral refutations with quick follow-ups~\cite{FGK05,GL03}. More recently, Allen, O'Donnell, Witmer~\cite{AOW15} gave a polynomial time algorithm for refuting random CSPs and Raghavendra, Rao and Schramm~\cite{RRS16} extended it to give sub-exponential algorithms with non-trivial guarantees. Both these algorithms are captured in the sum-of-squares semidefinite programming hierarchy. 

There's a long line of work proving lower-bounds for refuting random CSPs in proof complexity beginning with the work of Chv\'{a}tal and Szemer\'{e}di~\cite{CS88} on resolution refutations of random $k$-SAT. Ben-Sasson and Wigderson~\cite{BSW01, Ben01} and Ben-Sasson and Impagliazzo and Alekhnovich and Razborov~\cite{BI99, AR01a} further extended these results. 

One of the first lower bounds for random CSPs using SDP hierarchies was given by Buresh-Oppenheim~et~al.~\cite{BGH+03}; it showed that the \Lovasz--Schrijver$_+$ (LS$_+$) proof system cannot refute random instances of $k$-SAT with $k \geq 5$ and constant clause density $\Delta = m/n$.  Alekhnovich, Arora, and Tourlakis~\cite{AAT05} extended this result to random instances of $3$-SAT. The strongest results along these lines involve the Sum of Squares (AKA Positivstellensatz or Lasserre) proof system (see e.g.,~\cite{OZ13, BS14, Lau09,TCS-086} for surveys concerning SOS). Starting with Grigoriev~\cite{Gri01}, a number of works~\cite{Sch08,Tul09,BGMT12,OW14,TW13,MW16} have studied the complexity of SoS and the related Sherali-Adams and Lovasz-Schrijver proof systems for refuting random instances of CSPs. The strongest known results in this direction are due to the recent works~\cite{BCK15} and ~\cite{KMOW16} who gave an optimal 3-way trade-off between refutation strength, number of constraints and a certain complexity measure of the underlying predicate.  

Beyond semialgebraic proof systems and hierarchies, Feldman, Perkins, and Vempala \cite{FPV15} proved lower bounds for refutation of CSP$(P^\pm)$ using statistical algorithms when $P$ supports a $(t-1)$-wise uniform distribution.  Their results are incomparable to the above lower bounds for LP and SDP hierarchies: the class of statistical algorithms is quite general and includes any convex relaxation, but the \cite{FPV15} lower bounds are not strong enough to rule out refutation by polynomial-size SDP and LP relaxations. Very recently, Garg, Kothari, and Raz~\cite{DBLP:journals/corr/abs-2002-07235} gave lower-bounds for distinguishing planted from random CSPs in the bounded-memory model.

Semi-random models for various graph-theoreric optimization problems have been studied in \cite{blum-spencer,FK01,FK00,Coja-Oghlan07a,Coja-Oghlan07b,MMV12,MMT20} to mention just a few works. Semi-random models of the Unique Games CSP were studied in \cite{KMM}.
These works typically dealt with the optimization version of finding a good or planted solution in a semi-random instance.

\section{Overview of our approach} \label{sec:overview}
In this section, we give a high-level overview of the main ideas in our refutation algorithm by focusing on the case of $k$-XOR. 
This immediately implies refutation algorithms for all Boolean CSPs using Feige's standard ``XOR principle"~\cite{MR2121179-Feige02}. 

We will use the $\pm 1$ notation to express Boolean values, so that the XOR operation becomes the product. We can view a $k$-XOR instance $\phi$ as composed of a $k$-uniform hypergraph with $m$ edges specifying the tuples that are constraint constrained, a set of RHS's  $r \in \on^m$. We allow repeated edges in the hypergraph; this is necessary for the reduction from general semi-random CSP to XOR.

For an assignment $x \in \on^m$, define \[ \biasval_{\phi}(x) =\frac{1}{m} \sum_{j = 1}^m r_j \cdot \prod_{i \in C_j} x_i \ , \]
and let 
\[ \biasval_{\phi} := \max_{x \in \on^n} \biasval_\phi(x) \ . \]
Then, the fraction of constraints $\val_\phi(x)$ satisfied by an assignment $x$ equals 
\[ \val_\phi(x) = \frac{1+\biasval_{\phi}(x)}{2} \ . \]
A \emph{semi-random refutation} algorithm takes $H,r$ and outputs a value $v$ such that 
\begin{enumerate}
    \item  $v \geq \biasval_{\phi}$, and
    \item $\Pr_{r \in \on^m} [v \leq 2\epsilon] \geq 0.99.$\footnote{Observe that one can obtain a \emph{weak refutation} algorithm by simply running Gaussian elimination on the $k$-XOR system and outputting $1-1/m$ if the system is insoluble. However such algorithms do not yield even weak refutation algorithms for other CSPs and are not of interest to us in this work.} 
\end{enumerate}
Whenever $m \gg n/\epsilon^2$, there is an inefficient algorithm that succeeds in both these goals: the algorithm simply outputs $\biasval_{\phi}$. To obtain efficient refutation algorithms, we need to come up with polynomial-time computable upper bounds on $\biasval_{\phi}$ that are good approximations to $\biasval_{\phi}$ for semi-random $\phi$. 

All algorithms for refuting CSPs (including ours) are captured by a canonical semidefinite program (an appropriate number of levels of the sum-of-squares hierarchy). However, our analysis can be viewed as giving a stand-alone efficient algorithm that uses only a basic SDP and spectral norm bounds on appropriately (and efficiently) generated matrices. We will adopt this perspective in this overview. Our algorithms apply to settings where the constraint hypergraph $H$ is $k$-uniform but can have repeated hyperdges (i.e. is a $k$-uniform multi-hypergraph). This level of generality is required in order to obtain our semi-random refutation algorithm for all CSPs. However, for the purpose of the overview, we will restrict ourselves to simple $k$-uniform hypergraphs $H$. 

\subsection{Refuting $2$-XOR Instances}
For $k=2$, there are two natural strategies that yield optimal (up to absolute constants) refutation algorithms. 
We explain them in detail because one can view all known strong refutation algorithms as simply reducing to $2$-XOR and then using one of these two strategies. 

The key observation is that $\biasval_{\phi}(x)$ is a quadratic form: $\biasval_{\phi}(x) = x^{\top} A x$ where $A(i,\ell) = A(\ell,i) = r_j/2$ if $C_j = \{i,\ell\}$ for some $j$ and $0$ otherwise.
%
%
This immediately allows two natural, efficient upper-bounds for $\biasval_{\phi}(x)$.
\begin{enumerate}
\item \textbf{Spectral Method.} We observe that $\biasval_{\phi}(x) =x^{\top} A x \leq \norm{x}_2^2 \norm{A}_2$. Thus, bounding the spectral norm of $A$, which is efficiently computable, immediately gives an upper-bound on $\biasval_{\phi}$. When $\phi$ is random, $A$ is a sparse random matrix and known results~\cite{bandeira2016} imply that $\norm{A}_2 \leq O(\sqrt{\log n})$ with high probability.
This immediately yields a refutation algorithm that succeeds whenever $m\gg n \log n/\epsilon^2$. 

\item \textbf{Semidefinite programming.} This is a refinement of the above idea and relies on the following reasoning: $\biasval_{\phi} = \max_{x \in \on^n} x^{\top} A x \leq \max_{x,y \in \on^n} x^{\top} A y$. This latter quantity is known as the ``infinity to 1'' norm of $A$ and the well-known Grothendieck inequality shows that the value of a natural semidefinite programming relaxation (see Lemma~\ref{grothendieck_sdp})  $\sdp(\phi)$ approximates it within a factor $< 2$. It is easy to argue using the Chernoff bound followed by a union bound that $\max_{x,y \in \on^n} x^{\top} A y \leq O(\sqrt{n/m})$. Thus, $\biasval_{\phi} \leq \sdp(\phi) \leq O(\sqrt{n/m})$ with high probability giving a refutation algorithm that works whenever $m \gg n/\epsilon^2$. 
\end{enumerate}

\subsection{Refuting $k$-XOR for $k > 2$}
When $k>2$, a natural strategy is to simply find a ``higher-order" (tensor) analog of the Grothendieck inequality/eigenvalue-bound we used above. However, there are no tractable analogs of these bounds that provide a $O(1)$ approximation. Indeed, the best known algorithms~\cite{MR2448456-Khot08} lose $\poly(n)$ factors in approximation and such a loss might be necessary. As a result, all known efficient refutation algorithms rely on reducing $k$-XOR instances to $2$-XOR via a \emph{linearization} trick. 

\medskip\noindent\textbf{The case of even $k$.} For even $k$, this trick considers a $2$ XOR instance with ${n \choose {k/2}}$ variables, each corresponding to $k/2$-tuples of the original variables and each constraint being a $2$-XOR constraint in this new, linearized space. Observe that since when $m \gg n^{k/2}$, in the linearized space, we have $n^{k/2}$ variables and $\gg n^{k/2}$ constraints with independently random RHS. Using a simple argument based on Chernoff + union bound immediately shows that the resulting $2$-XOR instance has value at most $O(\sqrt{n^{k/2}/m})$. One can then use either of the above two methods to obtain a good approximation to this value and get a refutation algorithm. In particular, spectral refutation gives an algorithm that succeeds whenever $m \gg n^{k/2}(\log n)/\epsilon^2$ for \emph{random} $k$-XOR instances. This method, however, cannot yield any non-trivial refutation algorithm in the semi-random case.
Fortunately, the SDP method continues to work fine giving a semi-random refutation algorithm for even $k$ that succeeds whenever $m\gg n^{k/2}/\epsilon^2$. 

\medskip\noindent \textbf{The case of odd $k$ and the main challenge.} For odd $k$, the situation is more complicated since there is no natural ``balanced" linearization. The natural ``unbalanced" linearization yields a $2$-XOR instance on $n^{(k+1)/2}$ variables yielding a refutation algorithm that succeeds whenever $m \gg n^{(k+1)/2}/\epsilon^2$ that is off from our goal by a $\sqrt{n}$ factor. For the case when $H$ is \emph{random}, Barak and Moitra~\cite{BM15} found a clever way forward using their \emph{Cauchy-Schwarz trick}. This idea constructs a $2(k-1)$-XOR instance $\phi'$ with $n^{k-1}$ variables from the input $k$-XOR instance $\phi$, relates the value of $\phi'$ to $\phi$ and then refutes $\phi'$ (which is an even arity XOR instance) via linearization followed by the spectral method above. Specifically, $\phi'$ takes every pair of constraints in $\phi$ that share a variable and ``XORs'' them to derive a new, $2(k-1)$ XOR constraint. For $k=3$, this results in a $4$-XOR instance.

Notice that a priori, it is not clear that the $2$-XOR instance produced by linearizing the $2(k-1)$-XOR instance $\phi'$ produced via this route has a small true value! Indeed, the resulting constraints have \emph{dependent} right hand sides with only $m \approx n^{k/2}$ bits of randomness  while the number of variables in the linearized space is $n^{k-1}$. Thus, we cannot execute a Chernoff + union bound style argument to upper-bound the true value as we did for the even $k$ case above. Nevertheless, the proof of Barak and Moitra cleverly exploits the randomness in the constraint graph $H$ to prove this (via the spectral upper-bound) using a somewhat intense calculation using the trace method. 

Alas, this strategy is a non-starter in our case. The spectral bound that the technique of Barak and Moitra proves is simply not even true in the semi-random case, in general.
%
 Note that the true value of this $2$-XOR instance \emph{may} still be small---it is just that the spectral upper bound is too loose to capture it. Indeed, no known method (including ours) gives an upper bound on the value of this particular $2$-XOR instance in the semi-random case.

\medskip\noindent\textbf{Aside: Feige's semi-random refutation via combinatorial method.} It is instructive to note that Feige \emph{did} manage to obtain an efficient \emph{weak} refutation algorithm for semi-random\footnote{Feige's argument applies to a slightly more general ``smoothed" setting. We have not (yet) investigated if our approach also extends to this setting.}
$3$-SAT. As all refutation algorithms, Feige's method also relies on refuting semi-random $3$-XOR in order to accomplish this - precisely our main interest in this paper. However, his method skirts the above difficulties by giving up on the goal of \emph{strongly} refuting $3$-XOR and showing instead that for weak refutation of $3$-SAT, one only needs to weakly refute semi-random $3$-XOR with a value upper bound of $1- \Omega(1)/ \Paren{m-\tilde{O}(n^{3/2})}$. This latter bound can be obtain by a \emph{combinatorial method} that, at a high-level, involves breaking the $m$ constraints into groups of size $\approx n^{3/2}$ with an appropriate combinatorial structure (called \emph{even covers}) each and proving that each of them is individually unsatisfiable. 

Indeed, up until this work, the case of semi-random refutation of odd arity XOR, while occurring in several applications has remained unresolved. 


\medskip We now describe the key insights that lead to our method. While our overall algorithm could be thought of as simply analyzing the value of a natural SDP (see Section~\ref{refuting_semirandom_partitioned_2xor}), it is helpful to break open the analysis and view the algorithm as decomposing the instance and then using spectral method and SDP discussed above for appropriately chosen pieces. We adopt this perspective in this overview. We will focus on the case of $k=3$ to keep things simple here. 
\subsection{A matrix Bernstein based proof of random $3$-XOR refutation} 
In order to motivate our plan, it is helpful to give a quick and painless proof (that does not use the trace-moment method) of the spectral refutation for refuting \emph{random} $3$-XOR.  Our analysis will allow us to abstract out a key property of a random instance that causes our proof strategy to succeed. We will then use this property to define \emph{pseudo-random} instances and show that refuting an \emph{arbitrary} $3$-XOR instance can be reduced to refuting a pseudo-random instance and a collection of instances that have a certain usable structure. 

In addition to giving a simple proof of spectral refutation for random $3$-XOR\footnote{We note~\cite{DBLP:journals/corr/abs-1904-03858}  recently found a similar simple proof for the case of even arity \emph{random} XOR.}, this proof will also allow us to concretely explain both the decoupling and linearization steps.

An instance of $3$-XOR with $m$ constraints can be described by a size $n$ collection of  $n \times n$ matrices as follows.
Pick an arbitrary ordering on the $n$ variables. 
Concretely, define $B_{j,k} = B_{k,j} = r(i,j,k)/2$ if the constraint $x_i x_j x_k = r(i,j,k)$ appears in $\phi$ and $i <j$ and $i< k$ in the chosen ordering - in this case, we say that the constraint was ``colored'' by color $i$. Note that every constraint is colored by exactly one of the $n$ colors. We can thus write:
\begin{equation} \label{eq:1}
\biasval_{\phi}(x) = \frac{1}{m} \sum_{i = 1}^n x_i \cdot x^{\top} B_i x \mper
\end{equation}
\medskip\noindent\textbf{Decoupling.} Next, the idea is to apply Cauchy-Schwarz inequality on the RHS above and use $x_i^2 = 1$ for every $i$ to conclude that:
\[
 \biasval_{\phi}(x)^2 \leq \frac{1}{m^2} \sum_{i = 1}^n x_i^2 \sum_{i = 1}^n \Paren{x^{\top} B_i x}^2 = \frac{n}{m^2} \sum_{i = 1}^n \Paren{x^{\top} B_i x}^2 
 \mper
\]
Observe that one can view this step as ``decoupling'' the smallest variable (in the ordering chosen above) from the other two in every constraint. 

\medskip\noindent\textbf{Linearization.} Next, we observe that $\paren{x^\top B_i x}^2 = \paren{x^{\otimes 2}}^{\top} \Paren{ B_i \otimes B_i }x^{\otimes 2}$ where $x^{\otimes 2}$ is the $n^2$-dimensional vector indexed by ordered pairs of the $n$ variables and $x^{\otimes 2}(u,v) =x_u x_v$ for every $u,v \in [n]$. This allows us to conclude:
\begin{equation}
 \max_{x \in \on^n} \biasval_{\phi}(x)^2 \leq \frac{n}{m^2} \max_{y \in \on^{n^2}} \sum_{i = 1}^n y^{\top} B_i \otimes B_i y \leq \frac{n}{m^2} \norm{y}_2^2 \Norm{\sum_{i\leq n} B_i \otimes B_i} = \frac{n^3}{m^2}\Norm{\sum_{i\leq n} B_i \otimes B_i} 
 \label{eq:2-xor-reduction}\mper 
\end{equation}
Observe that we have arrived at a \emph{$2$-XOR} instance and a spectral upper-bound on its value. This is precisely the instance and the spectral upper-bound that the proofs of Barak and Moitra~\cite{BM15} (and Allen,O'Donnell and Witmer~\cite{AOW15}) generate. The spectral upper-bound is analyzed using a somewhat complicated trace-moment method based analysis in those works.

\medskip\noindent\textbf{Enter Matrix Bernstein.} Our analysis of the RHS of~\eqref{eq:2-xor-reduction} needs two simple observations: 
\begin{enumerate}
    \item the matrices $B_i \otimes B_i$ are independent for each $i$ and $\E B_i \otimes B_i = 0$,\footnote{Technically, $\E B_i \otimes B_i$ is non-zero due to the presence of a small (that amount to a lower-order term) of entries of non-zero expectation. See the proof of Lemma 5.7 for details.} and 
    \item $\max_{i \leq n} \Norm{B_i \otimes B_i} = \norm{B_i}^2 \leq O(1)$ by observing that a pair of variables can appear in at most $O(1)$ constraints with high probability. 
\end{enumerate}
We now want to apply the following standard matrix concentration inequality to upper-bound $\Norm{\sum_{i\leq n} B_i \otimes B_i}$.

\begin{fact}[Matrix Bernstein Inequality, \cite{Tropp2012} Theorem 1.4]
\label{fact:matrix-bernstein-symmetric}
Let $Z_1, Z_2, \ldots, Z_k$ be mean $0$, symmetric $d \times d$ matrices such that $\Norm{Z_i} \leq R$ almost surely. Then, for all $t \geq 0$, 
\[
\Pr\bigl[ \ \Norm{\sum_{i \leq k} Z_i} \geq t \  \bigr]\leq 2d\exp(-t^2/2(\sigma^2 + Rt/3))\mcom
\]
where $\sigma^2 = \Norm{\sum_{i \leq k} \E Z_i^2}$ is the norm of the matrix variance of the sum.
\end{fact}

To apply this inequality, we need to estimate the ``variance parameter'' $\sigma^2$. 
This is where we do need to do a simple but very helpful combinatorial observation. 

Let $\deg_i(v)$ be the number of constraints in $\phi$ that are colored by $i$ and contain variable $v$. Observe that for a random $3$-XOR instance with $m \ll n^2$, removing at most $o(m)$ constraints (or equivalently, losing an additional $o(1)$ in the certified upper-bound) ensures that every pair of variables appears in at most $1$ constraint in $\phi$. As a result, $\deg_i(v) \in \{0,1\}$ for every $i$ and $v$. 

Thus, each row of $B_i$ can have at most $1$ non-zero entry and as a result, the matrix $(B_i \otimes B_i)(B_i \otimes B_i)^{\top} = (B_i \otimes B_i)^2$ is a \textbf{diagonal} matrix with $(v,v')$ entry given by $\deg_i(v)\deg_i(v')$. Thus, the diagonal entry at say $(v,v')$ of the matrix $\sum_i (B_i \otimes B_i)^2$ equals $\sum_{i \leq n} \deg_i(v)\deg_i(v') \leq 1$! Applying Fact~\ref{fact:matrix-bernstein-symmetric} immediately yields that with probability at least $0.99$, $\Norm{\sum_i B_i \otimes B_i} \leq O(\log n)$. 

This immediately allows us to certify with probability at least $0.99$, a bound of $O(n^3 \log n/m^2) \ll \epsilon^2$ on $\biasval_{\phi}^2$ whenever $m \gg n^{3/2}\log n/\epsilon$. That finishes the proof!

\paragraph{Key lesson.} Here's the key observation in our analysis of the random case above: the matrices $B_i \otimes B_i$ are independent even when the constraint graph $H$ is worst-case! Thus, as long as we can control the ``variance parameter'' $\sigma^2$ - we can repeat the analysis from the random case above. 

\subsection{Semi-random refutation}
We now turn to describing our high level approach to accomplish the main goal in this paper, refuting semi-random instances of $k$-XOR.

\subsubsection{Springboarding off the random case}
Our approach for semi-random refutation is directly inspired from the above argument. First, we move to a variant of $3$-XOR that we call partitioned $2$-XOR via the following reasoning: we first ``color'' each constraint $C_i = (a,b,c) \in E$ by one of the three variables in it arbitrarily, say $a$. We then collect all the constraints of color $a$ in a group $G_a$. Note that $G_a$ is a graph of pairs of variables. Observe, then, that the following relationship immediately holds:
\[
\biasval_{\phi} = \max_{x \in \on^n}  \frac{1}{m}\sum_{i \leq m} r_i x_{C_i} \leq \max_{x,y \in \on^n} \frac{1}{m} \sum_{a \in [n]} \sum_{ C_j \setminus a \in G_a} r_j y_a x_{b,c}  \leq \max_{x \in \on^n} \frac{1}{m} \sum_{a \in [n]} \Abs{\sum_{ C_j \setminus a \in G_a} r_j x_{b,c}}\mper
\]

Note that the key difference from the usual $2$-XOR problem is that in the instance above,  we are interested in maximizing the sum of the \emph{absolute values} of the subinstance in each group. Such an instance is an example of the more general \emph{partitioned} $2$-XOR instance---all our algorithms apply at this level of generality---where there are multiple groups of $2$-XOR instances and we are interested in maximizing the $\ell_1$-norm of the value of the subinstances within each color group via some Boolean assignment.

To proceed, we identify the central feature of the partitioned $2$-XOR instance produced from a random $3$-XOR instance that led to the success of the above proof strategy: the \emph{butterfly degree} of a pair of variables (defined below). We use it to define the class of \emph{pseudorandom} constraint graphs $H$ as those with bounded butterfly degree. Observe that even though $\sigma$ depends on both the constraint graph and the ``right hand sides,'' the butterfly degree we define below (and used in the proof for the random case above) is a function only of the constraint graph.

\begin{definition}[Butterfly Degree and Pseudorandom Constraint graphs]
For any $v$, let $\deg_i(v)$ denote the number of constraints in the $i$-th $2$-XOR instance (in the input partitioned $2$-XOR instance) that include the variable $v$. For a pair of variables $v,v'$ the butterfly degree of $v,v'$ is $\frac{1}{n} \sum_{i \leq n} \deg_i(v) \deg_i(v')$. See Figure~\ref{fig:butterfly} for a visual depiction of this quantity.
A partitioned $2$-XOR instance is said to be \emph{pseudorandom} if the butterfly degree of every pair of variables is at most $O(1/\epsilon^2)$. 
\end{definition}

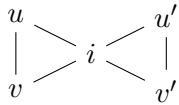
\begin{figure}[h]\label{butterfly}
\begin{center}
\begin{tikzpicture}
    \node[draw=none] (vp) at (1,-0.5) {$v'$};
    \node[draw=none] (up) at (1,0.5) {$u'$};
    \node[draw=none] (v) at (-1,-0.5) {$v$};
    \node[draw=none] (u) at (-1,0.5) {$u$};
    \node[draw=none] (i) at (0,0) {$i$};

    \draw (v) edge (u);
    \draw (vp) edge (up);
    \draw (v) edge (i);
    \draw (vp) edge (i);
    \draw (u) edge (i);
    \draw (up) edge (i);
\end{tikzpicture}
\caption{The butterfly degree counts diagrams of this form for fixed $v,v'$, where $(i,v,u)$ and $(i,v',u')$ each specify a unique constraint (e.g.\ $x_v x_u = y_i$).}
\label{fig:butterfly}
\end{center}
\end{figure}
A straightforward generalization of the argument from the previous subsection gives a refutation algorithm for all pseudorandom instances.

\subsubsection{Decomposition into bipartite $2$-XOR and structured $4$-XOR} 
When a partitioned $2$-XOR instance is \emph{not pseudorandom}, we observe that there must be $\deg_i(v)$ that are larger than the average for some $i,v$. Our main idea here is to partition the input instance into two parts: one where all $\deg_i(v)$'s are $O(1/\epsilon^2)$ and the other where some $\deg_i(v)$ is large. 

For the second part, we linearize to obtain a \emph{bipartite} $2$-XOR instance with one side containing $n$ variables (corresponding to $y$ above) and the other side containing $n^2$ variables. We then show this bipartite derived instance preserves the original value and thus, via the analysis based on the Grothendieck inequality above, we can use an SDP to efficiently refute this instance. 

The first part requires more work and is refuted by a carefully chosen sequence of spectral upper-bounds that utilize the Booleanity of the assignment to prove progressively tighter upper-bounds on the effective $\ell_2$ norm of the assignment. First, we use a carefully chosen weight for each of the $2$-XOR subinstances and then apply the Cauchy-Schwarz inequality. Intuitively, this is to ensure that the \emph{Cauchy-Schwarz} trick above is applied in the ``tightest'' possible parameter regime. These weights turn out to be almost equal in a partitioned $2$-XOR instance generated from a random $3$-XOR instance, up to a scaling factor. Thus, our matrix is essentially the same as the one employed above (and in~\cite{BM15}) when specialized to the random $3$-XOR setting but is, in general, different depending on the constraint graphs $H$.  This weighting does lead to an additional step of arguing that the value of the weighted instance can be related to the value of the unweighted instance. 

Next, we consider the matrix whose quadratic form captures the value of the above weighted $4$-XOR instance produced by the Cauchy-Schwarz trick. We decompose this matrix further into $\poly \log n$ pieces so as to populate rows with similar (up to constants) $\ell_2$ norm in piece. Our strategy is to use a spectral norm upper bound on each of this piece - however, a naive spectral norm bound does not suffice and is too large. Our key observation is that the upper-bound on the value of CSP instance so produced must be smaller than this estimate because the \emph{effective} $\ell_2$ norm of the assignment vector must be lower. This is because we can argue an appropriate geometrically decreasing upper-bound on the row-sparsity of the matrix in each piece of this decomposition.

\section{Preliminaries}
Below we present some notation we will use throughout this paper.
\begin{itemize}
\item For $n \in \mathbb{Z}_{\geq 1}$, define $[n] = \{1,2\dots, n\}$.
    \item A vector $v \in \mathbb{R}^n$ is \emph{Boolean} if all of its entries are either 1 or -1.
    \item The spectral norm $\|M\|$ is the maximum value of $x^\top M y$ over all unit vectors $x$ and $y$, or equivalently the maximum eigenvalue or singular value of $M$.
    \item The ``infinity to 1'' norm $\pmn{M}$ is the maximum value of $x^\top M y$ over all Boolean vectors $x$ and $y$.
    \item A multiset $S$ is a set which allows repeated elements.
    We think of identical multiset elements as ordered and distinguishable, so a function $f:S \to X$ may take different values on different copies of an element in $S$.
    Similarly, when performing an iterative operation over $S$ (for example, $\sum_{s \in S} f(s)$) we iterate over each copy of an element individually.
    If $S,T$ are multisets then $S+T$ is the disjoint union of their contents, so in particular $|S+T| = |S|+|T|$.
\end{itemize}

A \emph{$k$-XOR} instance $\phi = (V,E,r)$ is specified by a variable set $V$, a constraint \emph{multiset} $E$ of elements in $\binom{V}{k}$, and a sign function $r:E \to \{\pm 1\}$.
We call the instance \emph{semi-random} if the edge multiset can be arbitrarily chosen but $r$ is chosen uniformly at random from all functions $E \to \{\pm 1\}$ .
An assignment to $\psi$ is a vectors $x \in \{\pm 1\}^V$.
We say that a constraint $e \in E$ is satisfied when $x^e r(e) = 1$.
Define $\val_\phi(x)$ as the maximum fraction of satisfied constraints with assignment $x$, and define $\val_\phi = \max_{x}\val_\phi(x)$.

A \emph{partitioned 2-XOR} instance $\psi = (V,\ell,E,r)$ is specified by a variable set $V$, an integer $\ell$ denoting the number of parts, a constraint multiset $E$ of elements in $[\ell] \times \binom{T}{2}$, and a sign function $r:E \to \{\pm 1\}$.
We often write $r_i(e) = r(i,e)$.
We call the instance \emph{semi-random} if $E$ is arbitr ary and $r$ is chosen uniformly at random from all functions $E \to \{\pm 1\}$.
An assignment to $\psi$ is two vectors $x \in \{\pm 1\}^V$ and $y \in \{\pm 1\}^\ell$.
We say that a constraint $(i,e) \in E$ is satisfied by $(x,y)$ when $x^e y_i r_i(e) = 1$.
The set $E$ encodes a partition of a multiset of edges in $\binom{V}{2}$, where the $i$-th part may be realized explicitly as $T_i = \{e \in \binom{V}{2} \mid (i,e) \in E\}$.
Define $\val_\psi(x,y)$ as the maximum fraction of satisfied constraints with assignment $(x,y)$, and define 
\[ \val_\psi = \max_{x,y}\val_\psi (x,y) \ . \]
We now present some known results which will our proofs will rely on.
The first bound is a loose, but useful fact.
\begin{lemma}\label{gersh_rect}
	Let $X \in \mathbb{R}^{a \times b}$.
	Then the spectral norm of $X$ satisfies
	\[\|X\| \leq \max\left\{\max_{1 \leq i \leq a} \sum_{j=1}^b |X_{i,j}|, \max_{1 \leq j \leq b} \sum_{i=1}^a |X_{i,j}|\right\}.\]
	Equivalently, $\|X\|$ is upper bounded by the maximum $\ell_1$ norm of any row or column of $X$.
\end{lemma}
\begin{proof}
	This is a corollary of the Gershgorin circle theorem applied to the block matrix
	$\begin{bsmallmatrix}0 & X \\ X^\top & 0\end{bsmallmatrix}$,
	since the spectral norm of this block matrix equals the the spectral norm of $X$.
\end{proof}

The most important inequality underlying our work is the matrix Bernstein inequality.
The matrix Bernstein inequality is useful for bounding the spectral norm of random matrices which can be easily decomposed as a sum of matrices with uniformly bounded spectral norms.

\begin{theorem}[Rectangular Matrix Bernstein, \cite{Tropp2012} Theorem 1.6]\label{mbernstein_rect}
	Consider a finite sequence $\{Z_k\}$ of independent, random, matrices with common dimension $d_1 \times d_2$. Assume that each random matrix satisfies
	\[\E Z_k = 0\text{ and } \|Z_k\| \leq R\]
	almost surely.
	Define
	\[\sigma^2 = \max\left\{\left\|\sum_{k} \E[Z_k Z_k^\ast]\right\|, \left\|\sum_{k} \E[Z_k^\ast Z_k]\right\|\right\}\]
	Then, for all $t \geq 0$,
	\[\Psymb\bigl[ \Norm{\sum_{k} Z_k} \geq t\bigr] \leq (d_1 + d_2) \exp\left( \frac{-t^2/2}{\sigma^2 + Rt/3}\right).\]
\end{theorem}

 The following lemma is a standard result, relying on semidefinite programming and the Grothendieck inequality to bound $\pmn{M}$.
 
\begin{lemma}\label{grothendieck_sdp}
    Given a matrix $M$, one can efficiently approximate $\pmn{M}$ to within a constant factor $\leq 1.8$ of its true value.
\end{lemma}
\begin{proof}
    First note that the maximum value of $x^\top M y$ over Boolean vectors $x$ and $y$ exactly the maximum of $x^\top M y$ over vectors $x$ and $y$ with entries of magnitude at most 1.
    The Grothendieck inequality \cite{Grothendieck53} states that 
    \[\max_{\|u_i\|,\|v_j\| \leq 1} \sum_{i,j} M_{i,j} \langle u_i,v_j \rangle \leq K_G\max_{|x_i|,|y_j| \leq 1} \sum_{i,j} M_{i,j} x_i y_j\]
    where the left maximization is taken over vectors $u_1, \dots, u_n, v_1, \dots, v_n \in \mathbb{R}^d$ for some $d$, the right maximization is taken over numbers $x_1, \dots x_n, y_1, \dots, y_n \in \mathbb{R}$, and $K_G$ is an absolute constant independent of $M$, $n$, or $d$. It is known that $K_G \le \frac{\pi}{2 \ln (1+\sqrt{2})} < 1.8$~\cite{Krivine79}.
    The left side can be formulated and solved as a semidefinite program when $d = 2n$, and is clearly a relaxation of the case when $d=1$ which is what we aim to bound.
    The Grothendieck inequality shows that the upper bound obtained from this relaxation cannot be more than a factor of $K_G$ larger than $\pmn{M}$, so the solution to the SDP serves as a constant-factor approximation.
\end{proof}

\section{Reduction to Partitioned 2-XOR}
Our main technical result is the following theorem about refutation of semi-random partitioned 2-XOR.

\begin{theorem}[Semi-random Partitioned 2-XOR Refutation]\label{2xor_partitioned}
	Let $0 < \varepsilon < 1/2$.
	Let $\psi = (V,\ell,E,r)$ be a semi-random partitioned 2-XOR instance in which $|V| = n$ and $|E| = m$.
	Suppose $\ell = O(\operatorname{poly}(n))$.
	Then there is an absolute constant $K$ such that if
	\[m \geq \max\left\{\frac{Kn\sqrt{\ell} (\log n)^3}{\varepsilon^{5}}, \frac{\ell}{\varepsilon^2}\right\}\]
	we can with high probability efficiently certify that
	\[\val_\psi \leq 1/2 + \varepsilon.\]
\end{theorem}

As an application of this result, we prove the following theorem about refutation of semi-random $k$-XOR as a corollary, which is just a more detailed version of Theorem~\ref{thm:main-odd-xor}.

\begin{theorem}\label{kxor_refutation}
	Let $0 < \varepsilon < 1/2$.
	Let $\phi = (V, E, r)$ be a semi-random $k$-XOR instance in which $|V| = n$ and $|E| = m$.
	Then there is an absolute constant $K$ such that if
	\[m \geq \frac{Kn^{k/2}(\log n)^3}{\varepsilon^5}\]
	with high probability we can efficiently certify that
	\[\val_\psi \leq 1/2 + \varepsilon.\]
\end{theorem}
\begin{proof}
	Suppose $k$ is even.
	There is a natural relaxation of $\phi$ to a semi-random 2-XOR instance $\psi = (\binom{V}{k/2}, E', r')$ as follows. For every $k$-set $e \in E$, break it into two $(k/2)$-sets $e_1, e_2$ which union to $e$, and add the edge $(e_1, e_2)$ to $E'$ with sign $r'(e_1,e_2) = r(e)$.
	Relax $\psi$ to a partitioned 2-XOR instance $\psi'$ with a single part.
	Then Theorem~\ref{2xor_partitioned} with $\ell = 1$  and vertex set size $n^{\frac{k}{2}}$ implies $\psi'$ can be $(1/2 + \varepsilon)$-refuted, which also $(1/2 + \varepsilon)$-refutes $\phi$.
	We could alternatively refute semi-random 2-XOR without using the full power of Theorem~\ref{2xor_partitioned}, which we prove later as Theorem~\ref{semirandom_2xor}.

	Suppose $k$ is odd.
	We construct a partitioned 2-XOR instance $\psi = (\binom{V}{(k-1)/2}, \ell, E', r')$ with $\ell = n$ as follows.
	For every $e \in E$, break it arbitrarily into a single vertex $i$, and two disjoint $(k-1)/2$-sets $e_1$ and $e_2$ such that $i \cup e_1 \cup e_2 = e$.
	Let $E'$ be the multiset of all such $(i, (e_1,e_2))$.
	Define $r_i'(e_1,e_2) = r(i \cup e_1 \cup e_2)$.
	Say $x \in \{\pm 1\}^n$ is the optimal assignment for $\phi$.
	Let $x'$ be the $\pm 1$ vector indexed by $(k-1)/2$-sized subsets of $V$ such that $x_{S}' = x^S$, and let $y' = x$.
	Then
	\begin{align*}
	\val_\psi
	&\geq \frac{1}{m}\#\{(i,(e_1,e_2) \in E' \mid x_{e_1}' x_{e_2}' y_i' r_i'(e_1,e_2) = 1\} \\
	&=\frac{1}{m}\#\{(i,(e_1,e_2)) \in E' \mid x^{i \cup e_1 \cup e_2} r(i \cup e_1 \cup e_2) = 1\} \\
	&= \frac{1}{m}\#\{e \in E \mid x^{e} r(e) = 1\} \\
	&= \val_\phi.
	\end{align*}
	Thus it suffices to $(1/2 + \varepsilon)$-refute the partitioned 2-XOR instance $\psi$, which we can do by applying Theorem~\ref{2xor_partitioned} with $\ell = n$ partitions, and vertex set size $|V| = \binom{n}{(k-1)/2}$.
\end{proof}

\section{Refuting Semi-random Partitioned 2-XOR}\label{refuting_semirandom_partitioned_2xor}

The goal of this section is to prove Theorem~\ref{2xor_partitioned}, thus establishing Theorem~\ref{kxor_refutation}. Our algorithm itself can be compactly described as follows:

\begin{mdframed}
  \begin{algorithm}[Refuting Partitioned $k$-XOR] \label{refutation_algorithm}
    \label[algorithm]{alg:refuting-partitioned-k-xor}\mbox{}
    \begin{description}
    \item[Given:]
    A partitioned 2-XOR instance $\psi = (V, \ell, E, r)$ with $n$ variables and $m$ constraints.
    \item[Output:]
      Return either ``Not $(1/2+\varepsilon)$-satisfiable'' or ``Unknown''.
    \item[Operation:]\mbox{}
    \begin{enumerate}
    \item Let $\psi' = (V',E',r')$ be an empty semi-random 2-XOR instance.
    \item While there exists $i \in [\ell]$ and $v \in V$ such that there are at least $\Theta(\varepsilon^{-2})$ constraints of the form $(i,(v,u)) = C \in E$, for each such constraint $C$, set $V := V \cup \{(i,v), u\}$ and $E' := E' + \{((i,v), u)\}$ and remove $e$ from $E$. 
    
    \item Let $\mu_i(u,v)$ be the sum of $r(e)$ over all (possibly zero) duplicates of $e = (i,(u,v))$ in $E$, and let $t_i$ be the number of $e \in E$ with first coordinate $i \in [\ell]$.
    If $|E| \geq \varepsilon m/2$, verify with an SDP that
    \[\max_{\substack{x_{v,v'},y_{u,u'} \in \R^{n^2} \\ \|x_{v,v'}\| = \|y_{u,u'}\| = 1}} \sum_{(v,u) \neq (v',u') \in V^2} \sum_{i=1}^\ell \frac{\mu_i(v,u)\mu_i(v',u') }{\sqrt{t_i}}\langle x_{v,v'}, y_{u,u'}\rangle = O\left(\frac{\varepsilon^2 |E|^{3/2}}{\sqrt{l}}\right)\]
    and if this does not hold, return ``Unknown''.
    \label{refute_e1_step}
    
    \item Let $\mu'((i,v), u)$ be the sum of $r(e)$ over all (possibly zero) duplicates of $e = ((i,v), u)$ in $E'$.
    If $|E'| \geq \varepsilon m/2$, verify with an SDP that
    \[\max_{\substack{x_{i,v}, y_u \in \mathbb{R}^{|V'|} \\ \|x_{i,v}\| = \|y_u\| = 1}} \sum_{(i,v) \in V' \cap ([\ell] \times V)} \sum_{u \in V' \cap V} \mu'((i,v), u) \langle x_{i,v}, y_u \rangle = O(\varepsilon |E'|)\]
    and if this does not hold, return ``Unknown''.
    \label{refute_e2_step}
    
    \item If both verifications passed, return ``Not $(1/2+\varepsilon)$-satisfiable''.
  \end{enumerate}
    \end{description}    
  \end{algorithm}
\end{mdframed}

\subsection{Decomposing Semi-random Partitioned 2-XOR}
For a semi-random partitioned 2-XOR problem $\psi = (V,\ell,E,r)$, let $\deg_i(v)$ be the number of $(i,e) \in E$ such that $v \in E$.
The main idea of the proof of Theorem~\ref{2xor_bounded_partitioned} is to split the constraints of a general semi-random partitioned 2-XOR instance between two instances $\psi_1$ and $\psi_2$. The instance $\psi_2$ will satisfy $\deg_i(v) \leq d = O(\varepsilon^{-2})$ for all $i \in [\ell]$ and all $v \in V$. In the instance $\psi_1$, for every $(i,e) \in E$ is adjacent to some vertex $v$ with degree at least $d$ within its partition.
We use a degree-bounded semi-random partitioned 2-XOR refutation to bound $\psi_2$, and then relax the $\psi_1$ into a semi-random 2-XOR instance and refute this as well.
We establish the following two theorems with the aim of refuting $\psi_2$ and $\psi_1$ above.

\begin{theorem}[$d$-bounded Semi-random Partitioned 2-XOR Refutation]\label{2xor_bounded_partitioned}
	Let $0 < \varepsilon < 1/2$.
	Let $\psi = (V,\ell,E,r)$ be a semi-random partitioned 2-XOR instance with $|V| = n$, $\ell = O(\operatorname{poly}(n))$, and $|E| = m$.
	Suppose $\deg_i(v) \leq d = O(\poly(\varepsilon^{-1}, \log n))$ for all $i \in [\ell]$ and $v \in V$.
	Then there is an absolute constant $K$ such that if 
	\[m \geq \max\left\{\frac{Kd n\sqrt{\ell} (\log n)^3}{\varepsilon^{2}}, \frac{\ell}{\varepsilon^2}\right\}\]
	we can with high probability efficiently certify that
	\[\val_\phi \leq 1/2 + \varepsilon.\]
\end{theorem}

We will also use the following analysis of SDP based refutation whose proof we will later reproduce for completeness. This idea was used in the context of refutation for the first time in the work of Feige~\cite{Fei07}.

\begin{theorem}[Semi-random 2-XOR refutation] \label{semirandom_2xor}
    Let $0 < \varepsilon < 1/2$.
    Let $\phi = (V,E,r)$ be a semi-random 2-XOR instance with $|V| = n$ and $|E| = m$.
	Then there is an absolute constant $K$ such that for
	\[m \geq \frac{Kn}{\varepsilon^2}\]
	we can with high probability efficiently certify that
	\[\val_\phi \leq 1/2 + \varepsilon.\]
\end{theorem}

The proofs of Theorem~\ref{2xor_bounded_partitioned} and Theorem~\ref{semirandom_2xor} will be given later in this section.
For now, let us take them as facts and use them to prove Theorem~\ref{2xor_partitioned}.

\subsection{Combining $d$-bounded and $2$-XOR cases for general partitioned $2$-XOR refutation}

We start by iteratively defining two multisets $E_1$ and $E_2$ such that $E_1 + E_2 = E$.
Let $E_2$ initially be empty, and let $E_1$ initially be equal to $E$.
For $i \in [\ell]$ and $v \in V$, let $S(i,v) = \{u \in V \mid (i,\{v,u\}) \in E_2\}$.
Let $C$ be a fixed constant, whose exact value will be chosen later.
While there exist $i$ and $v$ such that $|S(i,v)| \geq C\varepsilon^{-2}$, remove everything in $S(i,v)$ from $E_1$ and add them all to $E_2$ instead.
Repeat this process until no such $i$ and $v$ exist. We have now split $E$ into two separate semi-random partitioned 2-XOR instances (defined by $E_1$ and $E_2$) over the same vertex set $V$, each of which we will refute separately.
Let $\psi_1 = (V, \ell, E_1, r)$ be the first instance and $\psi_2 = (V,\ell, E_2, r)$ be the second.

\begin{lemma}\label{2xor_refute_e1}
	If $|E_1| \geq \varepsilon m / 2$, we can $(1/2 + \varepsilon/2)$-refute $\psi_1$.
\end{lemma}

\begin{lemma}\label{2xor_refute_e2}
	If $|E_2| \geq \varepsilon m / 2$, we can $(1/2 + \varepsilon/2)$-refute $\psi_2$.
\end{lemma}

The proofs of Lemmas~\ref{2xor_refute_e1}~\&~\ref{2xor_refute_e2} (which we give later in this subsection) respectively describe Steps~\ref{refute_e1_step}~\&~\ref{refute_e2_step} of Algorithm~\ref{refutation_algorithm}. 
Using the above two lemmas, we can prove Theorem~\ref{2xor_partitioned}.

\begin{proof}[Proof of Theorem~\ref{2xor_partitioned}]
	Begin our procedure by pruning as above to obtain the sets $E_1$ and $E_2$.
	If the original instance $\psi$ has a $(1/2+\varepsilon)$-satisfying solution $(x,y)$, then across both instances $\psi_1$ and $\psi_2$ we can satisfy at least $(1/2+\varepsilon)m$ constraints by choosing the assignment $(x,y)$ for each instance.
	Let $s$ be the maximum number of constraints we can simultaneously satisfy in $\psi$, and let $s_j$ be the maximum number of constraints we can simultaneously satisfy in $\psi_j$.
	We give three cases, in each one bounding $s \leq s_1 + s_2$ by $(1/2 + \varepsilon)m$ to complete the proof.

	Suppose both $|E_1|,|E_2| \geq \varepsilon m / 2$.
	Then we may $(1/2+\varepsilon/2)$ refute both instances by Lemma~\ref{2xor_refute_e1}, and Lemma~\ref{2xor_refute_e2}, showing that
	\[s_1 + s_2 \leq (1/2 + \varepsilon/2)(|E_1| + |E_2|) \leq (1/2 + \varepsilon)m.\]

	Now suppose exactly one of $|E_1|$ and $|E_2|$ is at least $\varepsilon m / 2$, and the other has size below $\varepsilon m / 2$.
	Without loss of generality, let the smaller one be $E_2$.
	Refute the first instance, and note that $|s_2| \leq |E_2| \leq \varepsilon m / 2$.
	Then
	\[s_1 + s_2 \leq (1/2 + \varepsilon/2)|E_1| + \varepsilon m/2 \leq (1/2 + \varepsilon)m.\]

	There is no case when both of the $|E_j|$ are small, since they would comprise of at most $\varepsilon m$ total constraints.
\end{proof}

All that remains is to prove Lemma~\ref{2xor_refute_e1} and Lemma~\ref{2xor_refute_e2}.

\begin{proof}[Proof of Lemma~\ref{2xor_refute_e1}]
	Within each of the $\ell$ partitions $T_i$ of $\psi_1$, each of the $n$ vertices is adjacent to at most $C\varepsilon^{-2}$ edges because of the way we pruned $E_1$.
	Then in the style of Theorem~\ref{2xor_bounded_partitioned} we have a degree bound of $d = C\varepsilon^{-2}$.
	With a large enough constant in our lower bound on $m$, for any fixed constant $R$ we can force
	\[|E_2| \geq \varepsilon m / 2 \geq \frac{4RCn\sqrt{\ell} (\log n)^3}{\varepsilon^4} = \frac{Rd n\sqrt{\ell} (\log n)^3}{(\varepsilon/2)^2}.\]
	Therefore by Theorem~\ref{2xor_bounded_partitioned} with high probability we may $(1/2 + \varepsilon/2)$-refute the $d$-bounded partitioned 2-XOR instance $\psi_1$.
\end{proof}

\begin{proof}[Proof of Lemma~\ref{2xor_refute_e2}]
We reduce refutation of $E_2$ to refutation of (non-partitioned) semi-random 2-XOR.
Let $m' = |E_2|$.
Recall that we constructed $E_2$ by repeatedly adding sets $S(i,v)$ of size at least $K\varepsilon^{-2}$, where every element in $S(i,v)$ was of the form $(i, \{v,u\})$.
Let $X = \{(i,v) \mid S(i,v)$ was added to $E_2$ during its construction$\}$.
Since each set $S(i,v)$ is disjoint from all other such sets, there are $m'$ constraints in the $S(i,v)$ in total, and each $|S(i,v)| \geq C \varepsilon^{-2}$, it follows that $|X| \leq \frac{\varepsilon^2 m'}{C}$.

We a bipartite semi-random 2-XOR instance $\phi = (X \cup V, E_2', r')$ on the bipartitioned vertex set $X \cup V$ as follows.
Let
\[E_2' = \{((i,v), u) \mid (i, \{v,u\})\text{ was originally added to $E_2$ via }S(i,v)\}.\]
We can think the constraint graph of $\phi$ intuitively follows: For each edge $e \in T_i$ there is some $(i,v) \in X$ such that $v \in e$.
Then an edge $e \in T_i$ is transformed into an edge between $(i,v)$ and the vertex $u \in e$ which is not $v$.

We inherit the random sign function $r':E_1' \to \{\pm 1\}$ from $r$ by setting $r'((i,v),u) = r_i(v,u)$.
The distribution of $r'$ is uniform over all functions $E_2' \to \{\pm 1\}$, since $r'$ is just $r$ carried over to $E'$ via an isomorphism of sets.

We show now that refuting the 2-XOR instance $\phi$ is a relaxation of refuting the partitioned 2-XOR instance with constraint set $E_2$, so it suffices to refute the former.
Say $x,y$ is an maximally satisfying assignment for the partitioned 2-XOR instance with constraint set $E_2$.
Let $x'$ be a vector in $\{\pm 1\}^{X \cup V}$ defined as follows. Let $x_{(i,v)}' = y_i x_v$ for $(i,v) \in X$ and let $x_v' = x_v$ for $v \in V$.
Then for $e = (v,u)$ we have
\[x^{e} y_i r_i(e) = 1 \implies x_v y_i x_u r_{i,v}'(u) = 1 \implies x_v' x_{(i,v)}' r_{i,v}'(u) = 1\]
so the maximally satisfying assignment to $E_2'$ will satisfy at least as many constraints as the maximally satisfying assignment to $E_2$.

Let $n' = |X| + |V|$, the number of variables in $\phi$.
Then $n' \leq \varepsilon^2 m'/C + n \leq 2\varepsilon^2 m'/C$ since $\varepsilon^2 m'/C \geq \varepsilon^3 m/(2C) \geq \frac{Kn\sqrt{\ell}(\log n)^3}{2C\varepsilon^2} \geq n$ for sufficiently large $n$.
By choosing $C$ large enough, we can force $m' \geq \frac{Kn'}{(\varepsilon/2)^2}$ for any fixed constant $K$, forcing the conditions of Theorem~\ref{semirandom_2xor} to hold thereby allowing for $(1/2 + \varepsilon/2)$-refutation.
\end{proof}

\subsection{Semi-random 2-XOR Refutation}

Let us now reproduce the proof of the needed $2$-XOR refutation and complete the case handled by  Theorem~\ref{semirandom_2xor}.
\begin{proof}[Proof of Theorem~\ref{semirandom_2xor}]
Extend the domain of $r$ to all of $V^2$, so that previously undefined values are now 0.
Let there be a canonical pair representing each edge, so that at most one of $r(u,v)$ and $r(v,u)$ is nonzero.
Define $M \in \mathbb{R}^{V \times V}$ by $M(v,u) = r(v,u)$.
Then
\[(2\val_\phi-1)m = \max_{x \in \{\pm 1\}^V} x^\top M x \leq \pmn{M}\] so it suffices to certify that
\[\pmn{M} \leq 2 \varepsilon m.\]
It follows from the Lemma~\ref{grothendieck_sdp} that we can bound this maximum of $x^\top M y$ over Boolean vectors $x$ and $y$ to within a constant factor of its true value using an SDP.
Therefore, for refutation it suffices to show that the maximal value of $x^\top M y$ over Boolean vectors is actually below $C\varepsilon m$ with high probability for some sufficiently small constant $C$.

Fix $x,y \in \{\pm 1\}^n$.
Set $t = C\varepsilon m$.
By Bernstein's inequality,
\[\operatorname{P}\left[\sum_{(u,v) \in e} x_u y_v r(u,v) \geq t\right] \leq \exp\left(\frac{-t^2/2}{m + t/3}\right) \leq \exp(-C^2\varepsilon^2 m) \leq \exp(-C^2Kn)\]
Where the last inequality follows since $m \geq K n / \varepsilon^2$.
By choosing large enough $K$ we can force this value below $1/(p(n)2^{2n})$ for a polynomial $p(n)$ of arbitrarily high degree.
Then taking a union bound over all $2^{2n}$ choices of fixed $x,y \in \{\pm 1\}^n$ allows us to certify $x^\top M y \leq C \varepsilon m$ to any desired polynomial probability of success.
\end{proof}

\subsection{The $d$-Bounded Case: Proof of Theorem~\ref{2xor_bounded_partitioned}}

The remainder of this section (and the paper) is devoted to proving Theorem~\ref{2xor_bounded_partitioned} on the refutation of degree-bounded semi-random partitioned $2$-XOR.

Define the multiset $T_i = \{e \mid (i,e) \in E\}$, containing one copy of $e$ for each copy of $(i,e)$ in $E$.
Let $t_i = |T_i|$.
We assume that each $T_i$ is nonempty, as lowering $\ell$ only makes an instance easier to refute.

Define a function $\Phi: \{\pm 1\}^n \to \mathbb{R}$ as follows.
\[\Phi(x) = \sum_{i=1}^\ell \frac{1}{\sqrt{t_i}}\left(\sum_{e \in T_i} x^e r_i(e)\right)^2\]
The following lemma shows that when $\val_\psi$ has a large value, $\Phi(x)$ can take on a large value as well.
\begin{lemma}
	If $\val_\psi \geq 1/2 + \varepsilon$, then there is $x \in \{\pm 1\}^V$ such that
	\[\Phi(x) \geq \frac{4	\varepsilon^2 m^{3/2}}{\ell^{1/2}}\]
\end{lemma}
\begin{proof}
    Let $(x,y)$ be a maximally satisfying assignment to $\psi$.
	Let $s_i$ be the number of edges in $T_i$ which are satisfied by $(x,y)$.
	Note that $\sum_i s_i \geq (1/2 + \varepsilon)m$, $\sum_i t_i = m$, and
	$\sum_{e \in T_i} x^e r_i(e)y_i = 2s_i - t_i$.
	We use the fact that every $y_i^2 = 1$ to show
	\begin{align*}
	\Phi(x)
	&= \sum_{i = 1}^\ell \frac{1}{\sqrt{t_i}}\left(\sum_{e \in T_i} x^e r_i(e)\right)^2
	= \sum_{i = 1}^\ell \frac{1}{\sqrt{t_i}}\left(\sum_{e \in T_i} x^e r_i(e)y_i\right)^2
	= \sum_{i = 1}^\ell \frac{(2s_i - t_i)^2}{\sqrt{t_i}} \\
	&\geq \frac{\left(\sum_{i=1}^\ell 2s_i - t_i\right)^2}{\sum_{i=1}^\ell \sqrt{t_i}}
	\geq \frac{\big(2(1/2 + \varepsilon)m - m\big)^2}{\sqrt{\ell m}} = \frac{4\varepsilon^2 m^{3/2}}{\ell^{1/2}}
	\end{align*}
where the first inequality follows from Cauchy-Schwarz applied to the sequences $u_i = (2s_i-t_i)/t_i^{1/4}$ and $v_i = t_i^{1/4}$, and the second inequality follows from Cauchy-Schwarz applied to $\sum_i \sqrt{t_i}$, as well as our bounds on $\sum_i s_i$ and $\sum_i t_i$.
\end{proof}
So to certify $\val_\phi \leq 1/2 + \varepsilon$, it suffices to certify $\Phi(x) \leq \frac{4 \varepsilon^2 m^{3/2}}{\ell^{1/2}}$ for all $x$.

Notice that by focusing on bounding $\Phi(x)$, we have effectively removed $y$ from consideration.
Write
\[\Phi(x) = \left(\sum_{i=1}^\ell \frac{1}{\sqrt{t_i}}\sum_{(e, e') \in \binom{T_i}{2}} x^e x^{e'} r_i(e)r_i(e) \right) + \left(\sum_{i=1}^\ell \frac{1}{\sqrt{t_i}} \sum_{e \in T_i} (x^e)^2r_i(e)^2\right)\]
and let $\Phi_1(x)$ be the first term above and $\Phi_2(x)$ be the second.
We bound $\Phi(x)$ by bounding the two terms individually.
\begin{lemma}\label{2xor_phi2small}
We have the upper bound	\[\Phi_2(x) \leq \frac{\varepsilon^2 m^{3/2}}{\ell^{1/2}}  \ . \]
\end{lemma}
\begin{proof}
	See that every term in the inner summation in $\Phi_2(x)$ is 1, since $r_i(e),x^e \in \{\pm 1\}$.
	Thus the entire inner summation is $t_i$, meaning
	\[\Phi_2(x) = \sum_{i=1}^\ell \sqrt{t_i} \leq \sqrt{\ell \sum_{i=1}^\ell t_i} = \sqrt{\ell m}\]
	where the inequality follows from Cauchy-Schwarz.
	Since $\ell \leq \varepsilon^2 m$ by our lower bound on $m$, we see $\sqrt{\ell m} \leq \frac{\varepsilon^2 m^{3/2}}{\ell^{1/2}}$, completing the proof.
\end{proof}

Lemma~\ref{2xor_phi1small} is always true, so requires no certification.
Thus, if we can prove the following lemma, we will be done.

\begin{lemma}\label{2xor_phi1small}
	With high probability, we can certify that
	\[\Phi_1(x) \leq \frac{3\varepsilon^2 m^{3/2}}{\ell^{1/2}}\]
	for all $x \in \{\pm 1\}^n$.
\end{lemma}
The rest of this section is dedicated to proving Lemma~\ref{2xor_phi1small}.
For $i \in [\ell]$, define the random function $\mu_i : \binom{V}{2} \to \Z$ as the sum of $r_i(e)$ over all copies of $(i,e) \in E$.
Then define the matrix $M \in \mathbb{R}^{V^2 \times V^2}$ by
\begin{align*}
	M(v,v';u,u') =
	\begin{cases}
	\sum_{i=1}^\ell \frac{1}{\sqrt{t_i}}\mu_i(v,u)\mu_i(v',u') &\text{when $(v,u) \neq (v',u')$, and} \\
	0 &\text{otherwise.}
	\end{cases}
\end{align*}
Let $x^{\otimes 2}$ be the vector with $(x^{\otimes 2})_{u,v} = x_u x_v$.
By tracing definitions, we see that
\begin{equation*}
\Phi_2(x) = (x^{\otimes 2})^\top M (x^{\otimes 2}) \leq \pmn{M}
\end{equation*}
so if we can certify that $\pmn{M} \leq \frac{3\varepsilon^2 m^{3/2}}{\sqrt{\ell}}$ we are done.
By Lemma~\ref{grothendieck_sdp}, we can efficiently approximate this value within a factor $\leq 3$, so it suffices to just show that $\pmn{M}$ is at most $\frac{\varepsilon^2 m^{3/2}}{\sqrt{\ell}}$.

We start by assigning weights $\gamma(v,v')$ to all pairs $(v,v') \in V^2$.
Define
\begin{equation}\label{gamma_defn}\gamma(v,v') = \sum_{i=1}^\ell \frac{\deg_i(v) \deg_i(v')}{t_i}.\end{equation}
We refer to $\gamma(v,v')$ as the \emph{butterfly degree} of the pair $(v,v')$.
We can think of the butterfly degree as counting (weighted) diagrams of the form in Figure~\ref{fig:butterfly}, where a diagram in partition set $i$ is given weight $1/t_i$.
Intuitively, $\gamma(v,v')$ will be high when $v$ and $v'$ have relatively high degrees within relatively small $T_i$.
The next lemma shows that the total of the butterfly degrees of all pairs is not too large.
\begin{lemma}\label{2xor_gammasum}
We have	\[\sum_{v,v' \in [n]} \gamma(v,v') = 4m\]
\end{lemma}
\begin{proof}
Indeed,
\begin{align*}
	\sum_{v,v' \in [n]} \gamma(v,v')
	&=  \sum_{i=1}^\ell \frac{1}{t_i} \sum_{v,v' \in V} \deg_i(v) \deg_i(v')
	=  \sum_{i=1}^\ell \frac{1}{t_i} \left(\sum_{v \in V} \deg_i(v)\right)^2 \\
	&=  \sum_{i=1}^\ell \frac{1}{t_i} \left(2t_i\right)^2
	= 4 m \ . \qquad  \qedhere
	\end{align*}
\end{proof}

We now define a partition $S_0, S_1, \dots, S_{\log n}$ of $V^2$ by breaking pairs $(v,v')$ into different ``butterfly degree classes''.
We define two numbers $\alpha, \beta \in \mathbb{R}$ as follows:
\begin{align}
\label{eqn:alpha-beta}
&\alpha = \frac{Cd^2 \ell (\log n)^6}{\varepsilon^4 m},
&&\beta = \left(\frac{4 m}{\alpha}\right)^{1/\log n}
\end{align}
Where $C$ is some large enough absolute constant, with the precise meaning of ``large enough'' to be determined later.

We define the sets $S_j$ as follows.
\begin{align*}
	S_0 &= \left\{(v,v') \in V^2 \mid \gamma(v,v') \leq \alpha\right\} \text{ and} \\
	S_j &= \left\{(v,v') \in V^2 \mid \alpha\beta^{j-1} \leq \gamma(v,v') \leq \alpha \beta^{j}\right\}\text{ for $j \geq 1$.}
\end{align*}
If we are dealing with a purely random instance in which each constraint is drawn at random, $\gamma(v,v) = \widetilde{O}(1)$ for each pair $v,v'$, and $S_0 = V^2$.
We can think of $S_0$ as containing all pairs which behave pseudorandomly, while $S_j$ for larger $j$ contain pairs whose edge mass is larger and more heavily correlated.

In order for this partition to be legitimate, we need to establish that $\beta \geq 1$ and that the sets $\{S_j\}$ actually include all elements of $V^2$.
We will show now that $\beta$ is bounded on both sides by absolute constants.
\begin{lemma}\label{2xor_betasmall} The parameter $\beta$ defined in \eqref{eqn:alpha-beta} satisfies
	\[1 \leq \beta = O(1).\]
\end{lemma}
\begin{proof}
	We first establish the upper bound.
	Since each of the $n$ elements of $V$ is in at most $d$ pairs in each of the $\ell$ partitions, we know $m \leq \ell d n = O(\poly(n))$.
	Then expanding $\alpha$ and $\beta$ we have
	\[\beta =
	\left(\frac{4\varepsilon^4 m^2}{Cd^6\ell(\log n)^6}\right)^{1/\log n} \leq O\left( \ell n^2\right)^{1/\log n}=
	O(\poly(n))^{1/\log n} = O(1).\]

	Next we establish the lower bound of 1.
	\[
	\beta =
	\left(\frac{4\varepsilon^4 m^2}{Cd^6\ell(\log n)^6}\right)^{1/\log n}
	\geq \left(\frac{\varepsilon^2 m}{Cd^6(\log n)^6}\right)^{1/\log n} \geq 1^{1/\log n} = 1.
	\]
	where the first inequality holds since we assumed $\ell \leq \varepsilon^2 m$ via our lower bound on $m$, and the second inequality holds since $d = O(\poly(\varepsilon^{-1}, \log n))$.
\end{proof}

By Lemma~\ref{2xor_gammasum} and nonnegativity of $\gamma$, we know $\gamma(v,v') \leq 4 m = \alpha\beta^{\log n}$ always, so the $S_j$ do indeed cover all possible cases for $\gamma(v, v')$ and thus form a partition.

\begin{lemma}\label{2xor_sisize}
	$|S_0| \leq n^2$, and for $j \geq 1$,
	\[|S_j| \leq \frac{4m}{\alpha \beta^{j-1}}.\]
\end{lemma}
\begin{proof}
	The first bound follows as $S_0 \subseteq V^2$.
	For the second bound see that
	\[|S_j|\alpha \beta^{j-1} \leq \sum_{(v,v') \in S_j} \gamma(v,v') \leq \sum_{v,v' \in V} \gamma(v,v') = 4m\]
	where the last inequality follows from Lemma~\ref{2xor_gammasum}.
\end{proof}

The main idea of our refutation is as follows. A pair $(v,v')$ of high butterfly degree will bump up the value of $\Phi_1$.
However, Lemma~\ref{2xor_sisize} shows that there cannot be too many pairs contributing too large of a value.
To make this more precise, we will break $M$ down into a sum of $(1 + \log n)^2$ blocks $M_{j,k}$ indexed by $0 \leq j,k \leq \log n$.
Define $M_{j,k}$ as just the restriction of $M$ to rows in $S_j$ and columns in $S_k$.
Then clearly $x^\top M y = \sum_{j,k} x_{|S_j}^\top M_{j,k} y_{|S_k}$, so
\[\pmn{M} \leq \sum_{0 \leq j,k \leq \log n} \pmn{M_{j,k}}.\]
We will bound each $\pmn{M_{j,k}}$ individually via the inequality
\[\pmn{M_{j,k}} = \max_{x,y \text{Boolean}} x^\top M y \leq \sqrt{|S_j|\cdot |S_k|}\cdot\|M_{j,k}\|.\]

Intuitively, a similar bound will not work on $\pmn{M}$ directly because much of its mass can be unevenly distributed, causing the eigenvector associated to its maximal eigenvalue to look very far from a Boolean vector which results in a loose bound.
However, we might expect a spectral bound on the $M_{j,k}$ will work since by restricting to rows/columns in certain butterfly degree classes $S_j$ and $S_k$, we cause the vectors maximizing $x^\top M y$ to more closely resemble Boolean vectors.

As Lemma~\ref{2xor_sisize} gives a bound on the size of the sets $S_j$ and $S_k$, what remains is to bound $\|M_{j,k}\|$. 

\begin{lemma}\label{2xor_mnormsmall}
	For all $0 \leq j,k \leq \log n$, we have
	\[\|M_{j,k}\| = O\left(d\sqrt{\alpha \beta^{\max(j,k)}} (\log n)\right).\]
	with high probability.
\end{lemma}

The key ingredient in the proof of Lemma~\ref{2xor_mnormsmall} is the rectangular matrix Bernstein inequality, stated earlier as Theorem~\ref{mbernstein_rect}.
To apply it, we will write $M_{j,k}$ as a sum of independent, zero-mean matrices $B_{i,j,k}$ for $1 \leq i \leq \ell$, and provide upper bounds on both $\|B_{i,j,k}\|$ and the variance term $\sigma^2$ from the statement of the rectangular matrix Bernstein inequality.
We begin by defining the matrix $B_{i,j,k} \in \mathbb{R}^{S_j \times S_k}$ as
\begin{equation}
	B_{i,j,k}(v,v';u,u') =
	\begin{cases} \frac{1}{\sqrt{t_i}}\mu_i(v,u)\mu_i(v',u') &\text{when $(v,u) \neq (v',u')$, and} \\
	0 &\text{otherwise.}
	\end{cases}
\end{equation}
It follows from definitions that $\sum_{i=1}^\ell B_{i,j,k} = M_{j,k}$.
We now establish a uniform bound on the spectral norm of $B_{i,j,k}$.
\begin{lemma}\label{2xor_bibound}
	For all $0 \leq j,k \leq \log n$ and $1 \leq i \leq n$,
	\[\|B_{i,j,k}\| \leq d \sqrt{\alpha \beta^{\max(j,k)}}\]
\end{lemma}
\begin{proof}
    Fix $v,v' \in V$.
    Let $w_i = \deg_i(v) \deg_i (v')$.
    The $\ell_1$ norm of the row corresponding to $(v,v')$ in $B_{i,j,k}$ is at most $\frac{w_i}{\sqrt{t_i}}$, since each pair of edges $e,e' \in T_i$ with $v \in e$ and $v' \in e'$ contributes at most $\frac{1}{\sqrt{t_i}}$ to it, by the definition of $B_{i,j,k}$.
    Since the instance is $d$-bounded we know $w_i \leq d^2$, so $\frac{d}{\sqrt{w_i}} \geq 1$
    Thus
    \[\frac{w_i}{\sqrt{t_i}} \leq \left(\frac{d}{\sqrt{w_i}}\right) \cdot \frac{w_i}{\sqrt{t_i}} = d\sqrt{w_i/t} \leq d\sqrt{\sum_{i=1}^\ell w_i/t_i} = d\sqrt{\gamma(v,v')} \leq d\sqrt{\alpha \beta^j}\]
    where the last inequality follows since $(v,v') \in S_j$.
    This shows the $\ell_1$ norm of rows of $B_{i,j,k}$ are bounded by $d\sqrt{\alpha \beta^j}$, and an analogous argument works to show that the max $\ell_1$ norm of any column is bounded by $d\sqrt{\alpha \beta^k}$.
    Then by Lemma~\ref{gersh_rect} we are done.
\end{proof}

All that remains is to bound the variance term in the statement of Theorem~\ref{mbernstein_rect}.
We prove the following lemma.
\begin{lemma}\label{2xor_variancebound}
	For $0 \leq j, k \leq \log n$, we have
	\[\max\left\{\left\|\sum_{i=1}^\ell \E[B_{i,j,k} B_{i,j,k}^\top]\right\|,\left\|\sum_{i=1}^\ell \E[B_{i,j,k}^\top B_{i,j,k}]\right\| \right\}\leq 2 \alpha \beta^{\max(j,k)}\]
\end{lemma}

\begin{proof}
	Let $X_{j,k} = \sum_{i=1}^\ell \E[B_{i,j,k}B_{i,j,k}^\top]$.
	Because $B_{i,j,k} = B_{i,k,j}^\top$, the value we would like to bound is $\max\{\|X_{j,k}\|,\|X_{k,j}\|\}$.
	Then it suffices to show that $\|X_{j,k}\| \leq 2\alpha \beta^j$.

	We will first show that $X_{j,k}(v,v';u,u') \neq 0$ implies $\{v,v'\} = \{u, u'\}$.
	We expand the definition to find
	\begin{align*}
		X_{j,k}(v,v';u,u')
		&= \sum_{i = 1}^\ell\sum_{w,w'} \E[B_{i,j,k}(v,v';w,w')B_{i,j,k}(u,u';w,w')] \\
		&= \sum_{i = 1}^\ell \frac{1}{t_i} \cdot \sum_{\substack{(w,w') \in S_k \\ (v,w) \neq (v',w') \\ (u,w) \neq (u',w')}}\E[\mu_i(v,w)\mu_i(v',w')\mu_i(u,w)\mu_i(u',w')]
	\end{align*}
	For any $e,e' \in T_i$ either $\mu_i(e)$ and $\mu_i(e')$ are identical or $\mu_i(e)$ and $\mu_i(e')$ are independent with zero-mean.
	A term is nonzero precisely when the four $\mu(-)$ factors can be partitioned into two equal pairs.
	We cannot pair factors 1 and 2 since $\{v,w\} \neq \{v',w'\}$ by the condition on indices we are summing over.

	Suppose factors 1 and 3 are equal, and additionally factors 2 and 4 are equal. 
	Then we can see directly that $v = u$ and $v' = u'$.

	Suppose factors 1 and 4 are equal, and additionally factors 2 and 3 are equal.
	Then either (1) $v = u'$ and $w = w'$, or (2) $v = w'$ and $w = u'$.
	By looking at factors 2 and 3, case (1) implies $v' = u$, meaning $(v,v') = (u',u)$.
	Again looking at factors 2 and 3, case (2) implies $\{v,v'\} = \{u,u'\}$.

	Thus by considering all cases, we have determined that $X_{j,k}(v,v';u,u')$ is only ever nonzero when $\{v,v'\} = \{u,u'\}$.
	Let $D_i(v,w)$ be the number of copies of edge $(v,w)$ in partition set $i$.
	Then $\E[\mu_i(v,w)^2] = D_i(v,w)$ since it is a sum of $D_i(v,w)$ independent Rademacher random variables.
	In the case $v = u$ and $v' = u'$, we see
	\begin{align*}
	X_{j,k}(v,v';v,v')
	&= \ \sum_{i=1}^\ell \frac{1}{t_i} \cdot \sum_{\substack{(w,w') \in S_k \\ (v,w) \neq (v',w')}} \E[\mu_i(v,w)^2 \mu_i(v',w')^2] \\
	&= \ \sum_{i=1}^\ell \frac{1}{t_i} \cdot \sum_{\substack{(w,w') \in S_k \\ (v,w) \neq (v',w')}} D_i(v,w)D_i(v',w') \\
	&\leq  \ \sum_{i=1}^\ell \frac{\deg_i(v) \deg_i(v')}{t_i} \\
	&=  \ \gamma(v,v') \leq \alpha \beta^j
	\end{align*}
\noindent	where the last inequality follows since $(v,v') \in S_j$.
	When $v = u'$ and $v' = u$, we see 
	\begin{align*}
		X_{j,k}(v,v';v',v)
		&= \ \sum_{i = 1}^\ell \frac{1}{t_i} \cdot \sum_{\substack{(w,w') \in S_k \\ (v,w) \neq (v',w') \\ (v',w) \neq (v,w')}}\E[\mu_i(v,w)\mu_i(v',w')\mu_i(v',w)\mu_i(v,w')]. \\
\intertext{
    By inspecting the four factors in the expectation, and considering the two cases where they can be broken into two equal pairs, we see that the above is at most
}
		&\leq \ \sum_{i = 1}^\ell \frac{1}{t_i} \cdot \sum_{(w,w') \in V^2}\E[\mu_i(v,w)^2 \mu_i(v',w')^2] \\
		&= \gamma(v,v') \leq \alpha \beta^j
	\end{align*}
	where the last equality is because $(v,v') \in S_j$.
	Then each row and column of $X_{j,k}$ has at most two nonzero entries, whose total sum is bounded by $2\gamma(v,v') \leq 2\alpha \beta^j$.
	Then Lemma~\ref{gersh_rect} applies to show $\|X_{j,k}\| \leq 2\alpha \beta^j$, completing the proof.
\end{proof}

We now have enough information to apply the rectangular matrix Bernstein inequality to the decomposition $M_{j,k} = \sum_{i=1}^\ell B_{i,j,k}$.
By Theorem~\ref{mbernstein_rect}, Lemma~\ref{2xor_bibound}, and Lemma~\ref{2xor_variancebound}, we can set some $t = \Theta\left(d \sqrt{\alpha \beta^{\max(i,j)}} (\log n)\right)$ so that 
\begin{align*}
	\operatorname{P}\left[\|M_{j,k}\| \geq t\right]
	&\leq (|S_j| + |S_k|)\exp\left(\frac{-t^2/2}{2\alpha \beta^{\max(j,k)}+d\sqrt{\alpha \beta^{\max(j,k)}} t/3}\right) \\
	&\leq n2^{-\Theta(\log n)} \\
	& \leq \frac{1}{p(n)}
\end{align*}
for an arbitrarily high degree polynomial $p$.
This completes the proof of Lemma~\ref{2xor_mnormsmall}.
Since there are only $(1+\log n)^2$ different $M_{j,k}$ the probability that $\|M_{j,k}\| = O(\alpha \beta^{\max(i,j)})$ for all of them is at least $1-\frac{1}{p(n)}$ for a polynomial $p$ of arbitrarily high degree.
We now show that in the highly probable case that this bound holds for all $M_{j,k}$, $\Phi_2(x) \leq \frac{\varepsilon^2 m^{3/2}}{\sqrt{\ell}}$ for all $x \in \{\pm 1\}^n$.

We have the inequality
\begin{align}
\pmn{M} &= \sum_{0 \leq j,k \leq \log n} \pmn{M_{j,k}} \leq \max_{0 \leq j,k \leq \log n} 2  \pmn{M_{j,k}}(\log n)^2 \\
&\leq \max_{0 \leq j,k \leq \log n} 2 \sqrt{|S_j| \cdot |S_k|} \cdot \|M_{j,k}\|(\log n)^2. \label{m_bound}
\end{align}

We show that, no matter the choice of $j$ and $k$, the above bound is sufficiently small.
Since our upper bound on $|S_j|$ from Lemma~\ref{2xor_sisize} depends on whether $j = 0$ or not, we break into three cases.
We only consider the case when $j \leq k$, since $\pmn{M_{j,k}} = \pmn{M_{k,j}}$.

When the right side of Inequality (\ref{m_bound}) is maximized at $j = k = 0$, we have
\begin{align*}
2\sqrt{|S_0| \cdot |S_0|}\cdot\|M_{0,0}\|(\log n)^2
&\leq 2n^2\|M_{0,0}\|(\log n)^2 \tag{by Lemma~\ref{2xor_sisize}} \\
&= O\left(dn^2\sqrt{\alpha}(\log n)^3\right). \tag{by Lemma~\ref{2xor_mnormsmall}} 
\end{align*}
When the right side of Inequality (\ref{m_bound}) is maximized at $j = 0$ and $k \geq 1$, we have
\begin{align*}
2\sqrt{|S_0| \cdot |S_k|}\cdot\|M_{0,k}\|(\log n)^2
&\leq 2n\sqrt{\frac{4m}{\alpha \beta^{k-1}}}\|M_{0,k}\|(\log n)^2 \tag{by Lemma~\ref{2xor_sisize}} \\
&= O\left(n\sqrt{\beta} d\sqrt{m}(\log n)^3\right) \tag{by Lemma~\ref{2xor_mnormsmall}} \\
&= O\left(dn\sqrt{m}(\log n)^3\right). \tag{by Lemma~\ref{2xor_betasmall}}
\end{align*}

When the right side of Inequality (\ref{m_bound}) is maximized at $1 \leq j \leq k$, we have
\begin{align*}
2\sqrt{|S_j| \cdot |S_k|}\cdot\|M_{j,k}\|(\log n)^2
&\leq \frac{8m}{\alpha \beta^{(j+k-2)/2}} \|M_{j,k}\| (\log n)^2 \tag{by Lemma~\ref{2xor_sisize}}\\
&= O\left(\frac{d m (\log n)^3}{\sqrt{\alpha}\beta^{(j-2)/2}}\right) \tag{by Lemma~\ref{2xor_mnormsmall}}\\
&= O\left(\frac{d m (\log n)^3}{\sqrt{\alpha}}\right). \tag{by Lemma~\ref{2xor_betasmall}}\\
\end{align*}

Recall that we defined $\alpha = \frac{Cd^2 \ell (\log n)^6}{\varepsilon^4 m}$ for some large enough absolute constant $C$.
The third upper bound can be forced below any fixed constant fraction of $\frac{\varepsilon^2 m^{3/2}}{\sqrt{\ell}}$ by increasing the constant $C$ in the definition of $R$.
Additionally, the first two upper bounds can be forced below any fixed constant fraction of $\frac{\varepsilon^2 m^{3/2}}{\sqrt{\ell}}$ by increasing the constant $K$ in the lower bound on $m$ from the statement of Theorem~\ref{2xor_bounded_partitioned}.
This completes the proof of Theorem~\ref{2xor_bounded_partitioned}.


\phantomsection
 \addcontentsline{toc}{section}{References}
 \bibliographystyle{amsalpha}
 \bibliography{bib/custom,bib/references,bib/scholar,bib/mathreview,bib/dblp,bib/witmer}

\appendix

    

\end{document}